\documentclass[conference]{IEEEtran}

\usepackage[cmex10]{amsmath}

\usepackage{amssymb,amsthm}
\usepackage{tikz}
\usetikzlibrary{calc}
\usepackage{url,cite}
\usepackage[plain]{fancyref}

\newtheorem{theorem}{Theorem}
\newtheorem{proposition}{Proposition}
\newtheorem{corollary}{Corollary}
\newtheorem{definition}{Definition}
\renewcommand{\qed}{\hfill\IEEEQED}
\DeclareMathOperator*{\argmax}{\arg\!\max}

\IEEEoverridecommandlockouts

\begin{document}
\allowdisplaybreaks

\title{Capacity of the AWGN Channel with Random Battery Recharges}

\author{\IEEEauthorblockN{Dor Shaviv and Ayfer \"{O}zg\"{u}r}
\IEEEauthorblockA{Department of Electrical Engineering\\
Stanford University\\
Email: \{shaviv, aozgur\}@stanford.edu}
\thanks{This work was supported in part by the NSF CAREER award 1254786; by the Center for Science of Information (CSoI), an NSF Science and Technology Center, under grant agreement CCF-0939370; and by a Stanford Graduate Fellowship.}
}

\maketitle

\begin{abstract}
We consider communication over the AWGN channel with a transmitter whose battery is recharged with RF energy transfer at random times known to the receiver. We assume that the recharging process is i.i.d. Bernoulli. We characterize the capacity of this channel as the limit of an $n$-letter maximum mutual information rate under both causal and noncausal transmitter knowledge of the battery recharges. With noncausal knowledge, it is possible to explicitly identify the maximizing input distribution, which we use to demonstrate that the capacity with noncausal knowledge of the battery recharges is strictly larger than that with causal knowledge. We then proceed to derive explicit upper and lower bounds on the capacity, which are within 1.05 bits/s/Hz of each other for all parameter values.  

\end{abstract}

\IEEEpeerreviewmaketitle

\section{Introduction}

There has been significant recent progress in building wireless radios that possess no conventional batteries but are powered with wireless energy transfer, with latest developments reporting smaller device sizes, better harvesting efficiencies and increased communication ranges and data rates~\cite{Arb1,Arb2}. For example, the ant-sized radios of~\cite{Arb1} use the energy provided through the downlink channel in order to transmit over the uplink channel. We model communication with such externally powered transmitters by using a simple model. See \Fref{fig:channel}. Here a transmitter equipped with a battery of size $\bar{B}$ is communicating to a receiver over the AWGN channel. The transmitter's battery is recharged at random times: we assume that at each time $t$, the battery is recharged with probability $p$ independent of previous time instants. We assume that the recharging times are known either causally or noncausally both at the transmitter and the receiver. While it is natural for the transmitter to be (at least causally) aware  of its battery recharging times, the knowledge of the receiver is motivated by applications such as~\cite{Arb1,Arb2} where it is the receiver that powers the transmitter. 

The difficulty in characterizing the capacity of this setup lies in the fact that although the channel between the transmitter and the receiver is memoryless, the energy constraints on the transmitter lead to a random state for the system captured by the energy level of the battery. This state has memory and is input-dependent and even though the receiver can track the battery recharging times, the energy state of the battery, and therefore the state of the system, is known causally only at the transmitter but not at the receiver. We show that this state-dependent system is nevertheless equivalent to a conceptually simpler memoryless channel which we call the \emph{clipping channel}. The clipping channel admits real vectors as inputs, and outputs a clipped version of the input vector corrupted
by white Gaussian noise. The clipping length is random and follows a geometric distribution. The clipping channel is a memoryless channel with i.i.d. states.  Intuitively, each use of this channel corresponds to one \emph{epoch} over the original  channel with random battery recharges (RBR), where an epoch is the time period between two consecutive battery recharges. Using this equivalence, we provide an expression for the capacity of the RBR channel in \Fref{fig:channel}, and find an explicit formula for the approximate capacity of this channel which we show is within $1.05$ bits/channel use of the true capacity for all parameter values.

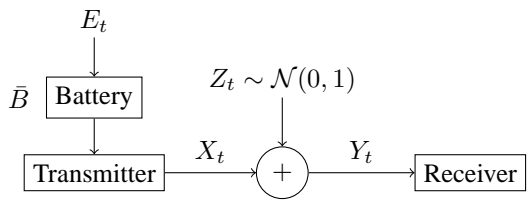
\begin{figure}[!t]
\centering
\begin{tikzpicture}
	\node[draw,rectangle] (Tx) at (0,0) {Transmitter};
	\node[draw,circle] (Sum) at (2.5,0) {$+$};
	\node[draw,rectangle] (Rx) at (5,0) {Receiver};
	\node[draw,rectangle] (Battery) at (0,1) {Battery};
	\node at (-1,1) {$\bar{B}$};
	
	\tikzstyle{every path}=[draw,->]
	\path (Battery) -- (Tx);
	\path (0,1.8) -- node[above,pos=0.1] {$E_t$} (Battery);
	\path (Tx) -- node[above] {$X_t$} (Sum);
	\path (Sum) -- node[above] {$Y_t$} (Rx);
	\path (2.5,1) -- node[above,pos=0.1] {$Z_t\sim\mathcal{N}(0,1)$} (Sum);
\end{tikzpicture}
\caption{System model.}
\label{fig:channel}
\end{figure}

This approximation is obtained by connecting the information-theoretic capacity of this channel to its long-term average throughput under optimal online power control. This latter communication-theoretic formulation of the problem assumes that there is an underlying transmission scheme which when allocated power $P_t$ at time time $t$ yields an instantaneous rate $\frac{1}{2}\log(1+P_t)$ and aims to maximize the long-term average throughput of the system by developing optimal online power control strategies~\cite{YangUlukus2012,Ozeletal2011,TutuncuogluYener2012,
online_infiniteB1,online_infiniteB2,online_infiniteB3,
online_infiniteB4,info2013}.
While only dynamic programming solutions are available for the online power control problem for  general energy harvesting systems~\cite{DP0,DP1,DP2,DP3}, we are able to explicitly characterize the optimal online power control strategy and the corresponding long-term average throughput for the specific model of interest here. This in turn yields an explicit formula for the approximate capacity of this system. In follow-up work, which we make simultaneously available on arXiv~\cite{DorISIT}, we were able to generalize the approach of the current paper to understand the capacity of energy harvesting systems powered by general i.i.d. energy arrival processes. While the results of~\cite{DorISIT} apply to the current case, we believe the results we develop in this paper specifically for the RBR channel are still of interest for a number of reasons: 1) The results we provide for this specific setup are stronger than those that follow from~\cite{DorISIT}. For example, while both~\cite{DorISIT} and the current paper characterize the capacity of the channel as an $n$-letter mutual information rate, the $n$-letter characterizations we obtain in this paper are much more explicit. This allows us, for example, to identify the maximizing input distribution for the $n$-letter expression when battery recharging times are known noncausally at the transmitter, which in turn allows us to show that noncausal knowledge of the battery recharges at the transmitter can strictly increase capacity over causal knowledge, even though the receiver also knows the i.i.d. battery recharging times. This result can be surprising given that for channels with i.i.d. states known both at the transmitter and the receiver, noncausal and causal knowledge of the states lead to the same capacity. Furthermore, we are able to explicitly solve the online power control problem for the RBR channel while~\cite{DorISIT} provides an approximately optimal online power control policy for the general case, which overall leads to an approximation of the capacity within 
$3.85$ bits/channel use, while the approximation in the current paper is within $1.05$ bits/channel use. 2) The generalization approach in~\cite{DorISIT} reveals the current setup as a canonical example for energy harvesting communication which is simple but yet captures most of the challenges involved in the general scenario. 3) Due to increasing interest in applications with external battery recharges, we believe the current setup can be of interest in its own right.

\subsection{Relation to Prior Work}
The setup we consider in this paper corresponds to a special case of the energy-harvesting communication channel, the capacity of which, despite significant recent interest \cite{OzelUlukus2012,Tutuncuogluetal2013,MaoHassibi2013,
DongOzgur2014,JogAnantharam2014,Ozeletal2014-2
}, remains an open problem.
In particular,~\cite{Tutuncuogluetal2013} considers a noiseless binary channel with a unit-sized battery where the battery recharges are known causally only at the transmitter. Our model resembles theirs in the fact that the energy arrival process is i.i.d. Bernoulli and each energy arrival fully recharges the battery. However, it is more general in the fact that we consider  noisy channels and battery size and the the input alphabet are not matched to each other. With a binary channel and unit battery, information can be only encoded in the timing of the unit-energy pulse which makes the setup of \cite{Tutuncuogluetal2013} equivalent to a timing channel. In our current case with an arbitrary battery size and continuous inputs, information can be encoded through real valued codewords and achieving capacity  requires to also devise an optimal power control strategy. The most closely related reference to our work is \cite{DongOzgur2014}~which considers an i.i.d. Bernoulli energy harvesting process where an energy packet of size $E$ is harvested with probability $p$ at each channel use and the transmitter is equipped with a battery of size $\bar{B}$ which can be either smaller or larger than $E$. (Reference~\cite{JogAnantharam2014} considers a special case of this model with $p=1$ in which case the harvesting process becomes deterministic).
Our model corresponds to a special case of the model in  \cite{DongOzgur2014} with $\bar{B}\leq E$. \cite{DongOzgur2014} provides upper and lower bounds on the capacity of this channel which are within 2.58 bits/s/Hz without providing an explicit expression for the capacity. To be more precise, 2.58 bits/s/Hz is the gap to capacity when the receiver has no information of the energy harvesting process. When the receiver has side information as we assume here, the gap can be readily decreased by $H(E_t)$, the entropy rate of the energy harvesting process, which is at most $1$ bit/s/Hz for Bernoulli arrivals. The contributions of the current paper with respect to \cite{DongOzgur2014} are: 1) we provide an explicit formula for the capacity by establishing an equivalence to the clipping channel; 2) we derive novel upper and lower bounds to the capacity in terms of a power control problem  for which we provide an explicit solution; this decreases the capacity approximation gap to 1.05 bits/s/Hz; 3) we show that the capacity with noncausal knowledge of the energy arrivals at the transmitter is strictly larger than the corresponding causal capacity.

\Fref{sec:definitions} describes our model for the channel with random battery recharges and \Fref{sec:main_results} contains the main results of the paper and the definition for the clipping channel.
\Fref{sec:proof_mainthm} provides the proof of our main theorem, namely the equivalence to the clipping channel.
In \Fref{sec:noncausal_strictly_greater} we show that noncausal observations of the energy arrivals can strictly increase capacity, and in \Fref{sec:bounds} we provide a derivation of capacity bounds.

\section{System Model}
\label{sec:definitions}

We consider a transmitter powered by RF energy transfer which communicates to a receiver over an AWGN channel, i.e. the output at time $t$ is $Y_t=X_t+Z_t$,
where $X_t\in\mathbb{R}$ is the input to the channel and $Z_t\sim\mathcal{N}(0,1)$ is the noise.
We assume that the transmitter has a battery with finite capacity~$\bar{B}$ which is recharged with probability $p$ at each channel use, i.e.
the energy arrivals $E_t$ are i.i.d. Bernoulli RVs:
\[
	E_t=
	\begin{cases}
		\bar{B}&\text{w.p. }p\\
		0&\text{w.p. }1-p.
	\end{cases}
\]
The effort to shrink down the size of wireless sensors and actuators limits the amount of energy that can be harvested at any given time, as well as the capacity of the storage unit that can be accommodated by the device. This necessitates the recharging process to operate at a scale comparable to the symbol duration~\cite{Arb1}. The randomness in the energy transfer process can be due to fluctuations in the alignment of antennas and the position of nodes, as well as randomness in the energy transfer times. We assume that the recharging times are known causally to both the transmitter and the receiver. The knowledge of the energy arrivals at the receiver is motivated by the fact that often it is the receiver that powers the transmitter, and the transmitter can acknowledge its battery exceeding a certain threshold by sending a short pulse to simplify operation. We also consider the case of noncausal energy arrival information at the transmitter, mainly for comparison with the causal case.
 
Under this model, energy of the channel input symbol at each time slot is limited by the available energy in the battery. Let $B_t$ represent the available energy in the battery at time $t$. The system energy constraints can be described as
\begin{align}
	|X_t|^2&\leq B_t, 	\label{eq:EH_constraint}\\
	B_t&=\min\{B_{t-1}-|X_{t-1}|^2+E_t,\bar{B}\}\label{eq:EH_battery}.
\end{align}
This implies that at time $t$, either $B_t=\bar{B}$ w.p. $p$, or $B_t=B_{t-1}-|X_{t-1}|^2$ w.p. $1-p$. 
We assume without loss of generality that $B_0=\bar{B}$,
which implies that we can also assume $E_1=\bar{B}$ w.p. 1.

An $(M,n)$ code for the random battery recharges (RBR) channel is a set of encoding functions
$f_t$ and a decoding function $g$:
\begin{align}
	f_t&:\mathcal{M}\times \mathcal{E}^t\to\mathcal{X}
		,\qquad t=1,\ldots,n,\label{eq:EH_encoding}\\
	g&:\mathcal{Y}^n\times\mathcal{E}^n\to\mathcal{M},
		\label{eq:EH_decoding}
\end{align}
where $\mathcal{E}=\{0,\bar{B}\}$, $\mathcal{X}=\mathcal{Y}=\mathbb{R}$
and $\mathcal{M}=\{1,\ldots,M\}$.
To transmit message $w\in\mathcal{M}$ at time $t=1,\ldots,n$, the transmitter sets $X_t=f_t(w,E^t)$.
The battery state $B_t$ is a deterministic function of $(X^{t-1},E^t)$, therefore also of $(w,E^t)$.
The functions $f_t$ must satisfy the energy constraint~\eqref{eq:EH_constraint}:
$|f_t(w,E^t)|^2\leq B_t(w,E^t)$.
The receiver sets $\hat{W}=g(Y^n,E^n)$.
The probability of error is
\[P_e^{\text{RBR}}=\frac{1}{M}\sum_{w=1}^{M}\Pr(\hat{W}\neq w\ |\ w\text{ was transmitted}).\]

The rate of an $(M,n)$ code is $R=\frac{\log M}{n}$.
A rate $R$ is achievable if for every $\varepsilon>0$ there exists a sequence of $(M,n)$ codes that satisfy $\frac{\log M}{n}\geq R-\varepsilon$ and $P_e^{\text{RBR}}\to0$ as $n\to\infty$.
The capacity $C_\text{RBR}$ is the supremum of all achievable rates.

When noncausal energy arrival information is available at the transmitter, the symbol transmitted at time $t$ can depend on the entire realization of the energy arrival process $E^n$. In this case, equation~\eqref{eq:EH_encoding} becomes
\[
	f_t:\mathcal{M}\times \mathcal{E}^n\to\mathcal{X}
		,\qquad t=1,\ldots,n,
\]
with remaining definitions unchanged.

\section{Main Results}
\label{sec:main_results}
The RBR channel described above has a random state $B_t$ which depends on the input and the exogenous energy arrival process. An explicit expression for the capacity is so far only available in terms of the Verd\'{u}-Han framework~\cite{MaoHassibi2013}.
However, we will see that this channel is conceptually equivalent to a clipping channel. In the sequel, we define a sequence of clipping channels parametrized by the length $N$ of the input vector.
For each $N$, the channel is memoryless and time-invariant, and lends to an almost trivial analysis of capacity.

\begin{definition}[Clipping Channel]
The $(N)$-clipping channel is a memoryless channel which at time $i$ admits inputs $\tilde{X}_i^{(N)}=(\tilde{X}_{i1}^{(N)},\ldots,\tilde{X}_{iN}^{(N)})\in\tilde{\mathcal{X}}^{(N)}=\mathbb{R}^{N}$ and outputs $\tilde{Y}_i^{(N)}=(\tilde{Y}_{i1}^{(N)},\ldots,\tilde{Y}_{iN}^{(N)})\in\tilde{\mathcal{Y}}^{(N)}=\mathbb{R}^{N}$. The inputs must satisfy the energy constraint 
\begin{equation}
	\|\tilde{X}_i^{(N)}\|^2=\sum_{j=1}^N|\tilde{X}_{ij}^{(N)}|^2\leq\bar{B}.
	\label{eq:clp_constraint}
\end{equation}
Each use of the channel is associated with a state variable $L_i$, called the clipping length. $L_i$ are i.i.d. RVs, independent of the input,
and follow a geometric distribution with parameter $p$:
\[ \Pr(L_i=k)=(1-p)^{k-1}p,\qquad k=1,2,\ldots \]
The states $L_i$ are known at the receiver but not at the transmitter.
At channel use $i$, if $L_i\leq N$, the channel output is given by
\[
	\tilde{Y}_{ij}^{(N)}=
	\begin{cases}
		\tilde{X}_{ij}^{(N)}+Z_{ij}&,j\leq L_i\\
		0&,j>L_i
	\end{cases}
\]
where $Z_{ij}\sim\mathcal{N}(0,1)$ are i.i.d. for different $i,j$, independent of $\tilde{X}_i^{(N)}$ and $L_i$. If $L_i>N$, the channel outputs $\tilde{Y}_i^{(N)}=0$.

An $(M,n)$ code for the clipping channel consists of encoding and decoding functions
\begin{align}
	\tilde{f}_i^{(N)}&:\mathcal{M}\to \tilde{\mathcal{X}}^{(N)},\qquad i=1,\ldots,n,
		\label{eq:clp_encoding}\\
	\tilde{g}^{(N)}&:(\tilde{\mathcal{Y}}^{(N)})^n\times\mathcal{L}^n\to\mathcal{M},
		\label{eq:clp_decoding}
\end{align}
where 
$\mathcal{L}=\mathbb{N}$,
and $\mathcal{M}=\{1,\ldots,M\}$.
The transmitted codeword is $\tilde{X}_i^{(N)}=\tilde{f}^{(N)}_i(w)$, $i=1,\ldots,n$, and the encoding functions must satisfy
$\|\tilde{f}_i^{(N)}(w)\|^2\leq\bar{B}$.
The decoded message is $\hat{W}=\tilde{g}^{(N)}((\tilde{Y}^{(N)})^n,L^n)$. The rate of the code is $R=\frac{\log M}{n}$ and the capacity $C_\text{clp}^{(N)}$ is defined in the standard way as the supremum of all achievable rates.
\label{def:clipping}
\end{definition}

Intuitively, at each channel use the channel chooses an i.i.d. clipping length $L_i$ and outputs only the first $L_i$ components of the input vector under additive white Gaussian noise (The case $L_i>N$ and the corresponding behaviour of the channel are rather technicalities; in the sequel we will be interested in $N\to\infty$ in which case the length of the input vector goes to infinity and the probability of $L_i>N$ goes to zero).
The $(N)$-clipping channel is a standard vector memoryless channel with i.i.d. state information $L_i$ available at the receiver. The capacity of this channel is well known and is given by
\begin{align*}
	C_\text{clp}^{(N)}
	&=\max_{\substack{p(\tilde{x}^{(N)}):\\ \|\tilde{X}^{(N)}\|^2\leq\bar{B}}} I(\tilde{X}^{(N)};\tilde{Y}^{(N)},L)\\
	&=\max_{\substack{p(\tilde{x}^{(N)}):\\ \|\tilde{X}^{(N)}\|^2\leq\bar{B}}} I(\tilde{X}^{(N)};\tilde{Y}^{(N)}|L).
\end{align*}
We can rewrite this expression in the following explicit form
\begin{align}
	C_\text{clp}^{(N)}
	&=\max_{\substack{p(\tilde{x}^{(N)}):\\ \|\tilde{X}^{(N)}\|^2\leq\bar{B}}} \sum_{k=1}^{N}p(1-p)^{k-1}I(\tilde{X}^{(N)};\tilde{Y}^{(N)}|L=k)\nonumber\\
	&=\max_{\substack{p(x^N):\\ \|X^N\|^2\leq\bar{B}}} 
		\sum_{k=1}^{N}p(1-p)^{k-1}I(X^k;X^k+Z^k),
		\label{eq:clp_capacity}
\end{align}
where in the last line $X^N=(X_1,\ldots,X_N)$ (and $X^k=(X_1,\ldots,X_k)$), and $Z^N=(Z_1,\ldots,Z_N)$ is a vector with  i.i.d. Gaussian entries, i.e., $Z_i\sim\mathcal{N}(0,1)$ i.i.d for $i=1,\dots,N$.

It is easy to see that $C_\text{clp}^{(N)}$ is monotonically increasing in $N$ and is also bounded above (see also Proposition~\ref{prop:bounds}), therefore it has a limit, which we call $C_\text{clp}$:
\begin{equation}
	C_\text{clp}\triangleq\lim_{N\to\infty}C_\text{clp}^{(N)}
	=\sup_{N\geq1}C_\text{clp}^{(N)}.
\end{equation}
Intuitively, taking $N\to\infty$ gives a channel with infinitely long input and output, thus simulating an epoch in the original channel, i.e. the time interval between two successive battery recharges. Indeed, we show that $C_\text{clp}$ is in fact, up to a constant factor, the capacity of the RBR channel $C_\text{RBR}$. 
We bring the following theorem without proof:
\begin{theorem}[Channel Equivalence]
\label{thm:equivalence}
\[ C_\text{RBR}=p\cdot C_\text{clp}. \]
\end{theorem}
\begin{IEEEproof}
See \Fref{sec:proof_mainthm}.
\end{IEEEproof}
Although the two channels are clearly related, with each use of the clipping channel corresponding to one epoch over the RBR channel, the fact that they are equivalent may be a priori unclear. Indeed, these two channels have quite different characteristics: the first has an input-dependent state with memory which is causally known at the transmitter; the second is a simple memoryless channel with states unknown at the transmitter. The intuitive connection is that whenever there
is a battery recharge in the first channel, the system \emph{resets}, and any memory of
the channel (which is embedded in the state of the battery) is erased. However, even with this intuition, it may be unclear why one could not benefit from transmitting the symbols in each epoch in a sequential manner (having as side information the time since the last battery recharge) as compared to one-shot transmission of the ``epoch symbol'' with no side information. The proof of the theorem formally argues that codes designed for one channel can be used over the other channel with similar performance.  
 
Using Theorem~\ref{thm:equivalence} and~\eqref{eq:clp_capacity}, we obtain the capacity of the AWGN channel with random battery recharges:
\begin{corollary}[Capacity of the RBR Channel]
\label{cor:capacity}
The capacity of the channel defined in \Fref{sec:definitions} is given by
\begin{equation}
	C_\text{RBR}=\lim_{N\to\infty}\max_{\substack{p(x^N):\\ \|X^N\|^2\leq\bar{B}}}
		\sum_{k=1}^{N}p^2(1-p)^{k-1}I(X^k;X^k+Z^k).
\label{eq:EH_capacity}
\end{equation}
\end{corollary}


It is easy to extend this result to the case of \emph{noncausal} observations of the energy arrival process.
In this case, it can be shown that the channel is equivalent to a clipping channel with state $L_i$ available at both the transmitter and the receiver.
The capacity of such a channel is also standard and is obtained by optimizing over all input distributions conditioned on the state:
\[
	C_\text{clp,noncausal}^{(N)}=
	\max_{\substack{p(\tilde{x}^{(N)}|l):\\ \|\tilde{X}^{(N)}\|^2\leq\bar{B}}} I(\tilde{X}^{(N)};\tilde{Y}^{(N)}|L).
\]
Using the equivalence and writing the above expression explicitly (in a form analogous to \eqref{eq:clp_capacity}) we get the following result:
\begin{theorem}[Noncausal Capacity]
\label{thm:capacity_noncausal}
The capacity of the channel defined in \Fref{sec:definitions} with energy arrival information available noncausally at the transmitter and the receiver is given by
\begin{equation}
\label{eq:EH_capacity_noncausal}
	C_\text{RBR,noncausal}=
	\sum_{k=1}^{\infty}p^2(1-p)^{k-1}
	\max_{\substack{p(x^k):\\ \|X^k\|^2\leq\bar{B}}}
	I(X^k;X^k+Z^k).
\end{equation}
\end{theorem}
\begin{proof}
See Appendix~\ref{sec:noncausal_capacity_proof}.
\end{proof}
It is possible to explicitly identify the maximizing input distribution in the above expression by using the results of~\cite{Smith1971,ShamaiBarDavid1995,
ChanHranilovicKschischang2005}, which characterize the capacity of amplitude-constrained channels. In particular, \cite{ChanHranilovicKschischang2005} shows that the maximizing $X^k$ in \eqref{eq:EH_capacity_noncausal} is distributed over a finite set of $k$-dimensional spheres with uniform phase, where the number of spheres is determined by the value of $\bar{B}$ (ex. when $\bar{B}$ is very small, $X^k$ is uniformly distributed over a single sphere of radius $\sqrt{\bar{B}}$). Using this result, we suggest the following proposition.
\begin{proposition}
\label{prop:noncausal_strictly_greater}
Noncausal observations of the energy arrival process strictly increase capacity. That is,
\[ C_\text{RBR,noncausal}>C_\text{RBR}. \]
\end{proposition}
\begin{IEEEproof}
See \Fref{sec:noncausal_strictly_greater}.
\end{IEEEproof}
This result may be surprising given that for a memoryless channel with i.i.d. state $S$, the capacity with side information at both the transmitter and the receiver is given by $I(X;Y|S)$, whether the side information is available causally or noncausally. The difference here is that even though the battery recharges $E_t$ are i.i.d. and known to both the transmitter and the receiver, the state of the system is captured by $B_t$ rather than $E_t$, which has  memory and is unknown to the receiver due to its input-dependence. 
The fact that noncausal knowledge of the energy arrivals strictly increases capacity
can be also observed by using the upper and lower bounds on the causal and noncausal capacities developed below.

Despite being relatively simpler than previous results\footnote{
In~\cite{ShavivNguyenOzgur2015}, we characterize the capacity of the general energy harvesting channel
in the form 
\[C=\lim_{n\to\infty}\frac{1}{n} \sup I(U^n;Y^n),\]
where the domain of the optimization problem is suitably defined. Note that the capacity expressions in \eqref{eq:EH_capacity} and \eqref{eq:EH_capacity_noncausal} are much more explicit, and in particular, it is this explicit form that allows us to identify the maximizing input distribution in \eqref{eq:EH_capacity_noncausal}.}, \eqref{eq:EH_capacity} and \eqref{eq:EH_capacity_noncausal} are difficult to compute explicitly. In particular,  \eqref{eq:EH_capacity} is a multi-letter expression that involves optimization over an infinite dimensional space. Therefore, we wish to find suitable approximations. More specifically, we provide an upper and a lower bound, separated by a constant gap of approximately 1.05 bits:
\begin{proposition}[Capacity Bounds]
\label{prop:bounds}
The capacity of the RBR channel is bounded by:
\begin{equation}
\label{eq:bounds}
	\bar{C}-\frac{1}{2}\log\left(\frac{\pi e}{2}\right)
	\leq C_\text{RBR}
	\leq \bar{C},
\end{equation}
where
\begin{equation}
\label{eq:def_Cbar}
	\bar{C}\triangleq
	\lim_{N\to\infty}
	\max_{\substack{\{\mathcal{E}_i\}_{i=1}^{N}:\\
					\mathcal{E}_i\geq0\ ,i=1,\ldots,N\\
					\sum_{i=1}^{N}\mathcal{E}_i\leq\bar{B}}}
	\sum_{i=1}^{N}p(1-p)^{i-1}\frac{1}{2}\log(1+\mathcal{E}_i).
\end{equation}
\end{proposition}
\begin{IEEEproof}
See \Fref{sec:bounds}.
\end{IEEEproof}
It can be shown that the upper bound $\bar{C}$ in \eqref{eq:def_Cbar} corresponds to the online power control problem, extensively studied in the literature in the general framework of energy-harvesting channels \cite{YangUlukus2012,TutuncuogluYener2012,Ozeletal2011}. Here, one assumes that there is an underlying transmission scheme
operating at a finer time-scale, such that allocating power $P$ to
this scheme yields an information rate $r(P) = \frac{1}{2} \log(1 + P)$,
and focuses on the optimal power allocation policy satisfying the energy constraints on the transmitter. For the specific channel of interest here, this online power control problem can be explicitly solved. In particular, we apply the KKT conditions to the optimization problem in~\eqref{eq:def_Cbar},
to obtain the optimal values of $\mathcal{E}_i$
(see Appendix~\ref{sec:kkt_solution}):
\begin{equation}\label{eq:optei}
	\mathcal{E}_i=\begin{cases}
		(\tilde{N}+\bar{B})\frac{p(1-p)^{i-1}}{1-(1-p)^{\tilde{N}}}-1
			&,i=1,\ldots,\tilde{N}\\
		0&,i>\tilde{N}
	\end{cases}
\end{equation}
where $\tilde{N}$ is the smallest positive integer satisfying 
\[ 1>(1-p)^{\tilde{N}}[1+p(\bar{B}+\tilde{N})]. \]
This gives the following expression for $\bar{C}$:
\begin{align*}
	\bar{C}&=
	\frac{1-(1-p)^{\tilde{N}}}{2}\log\left(
		\frac{p(\bar{B}+\tilde{N})}{1-(1-p)^{\tilde{N}}}\right)\\*
	&\qquad{}+\frac{1-p-(1-p)^{\tilde{N}}(1-p+\tilde{N}p)}{2p}\log(1-p).
\end{align*}
Combined with \eqref{eq:bounds}, this is the capacity of the RBR channel within 1.05 bits/channel use.

It was shown in~\cite{OzelUlukus2012} that the capacity of an energy harvesting channel
with infinite battery size is $\frac{1}{2}\log(1+\mathbb{E}[E_t])$.
Clearly, this is an upper bound to the capacity of our channel, and this can be readily obtained 
from the result of Proposition~\ref{prop:bounds}.
Using concavity of the log function in~\eqref{eq:def_Cbar}:
\begin{align*}
	\bar{C}
	&\leq \lim_{N\to\infty}
		\max_{\substack{\{\mathcal{E}_i\}_{i=1}^{N}:\\
						\mathcal{E}_i\geq0\ ,i=1,\ldots,N\\
						\sum_{i=1}^{N}\mathcal{E}_i\leq\bar{B}}}
		\frac{1}{2}\log\left(1+\sum_{i=1}^{N}p(1-p)^{i-1}\mathcal{E}_i\right)\\
	&= \frac{1}{2}\log(1+p\bar{B}),
\end{align*}
where the last step follows because the optimal values for the first line are $\mathcal{E}_1=\bar{B}$ and 
$\mathcal{E}_i=0$ for $i\geq2$.
\cite{DongOzgur2014} used this upper bound corresponding to infinite battery size to bound the capacity of the energy harvesting channel with Bernoulli energy arrivals. \Fref{fig:bounds} illustrates that the upper bound we provide here is strictly smaller than the infinite battery upper bound. Similarly, our lower bound here is based on the optimal power allocation strategy we characterize in \eqref{eq:optei}, while the lower bound in \cite{DongOzgur2014} is based on a suboptimal power allocation policy. 

\begin{figure}[!t]
\centering
	\includegraphics[width=2.5in]{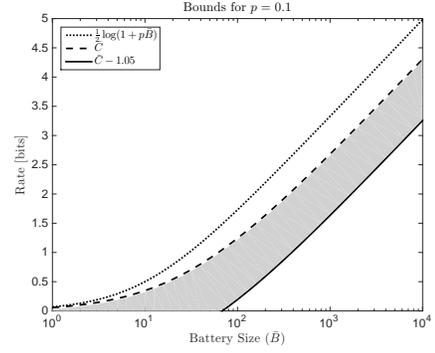}
\caption{Upper and lower bounds for $p=0.1$. The shaded region indicates where the capacity of the energy harvesting channel can lie.}
\label{fig:bounds}
\end{figure}

Similar bounds can be obtained for~\eqref{eq:EH_capacity_noncausal}, which we state in the following proposition.
\begin{proposition}
\label{prop:bounds_noncausal}
The capacity of the RBR channel with noncausal energy arrival information is bounded by:
\begin{equation}
\label{eq:bounds_noncausal}
	\bar{C}_\text{noncausal}
		-\frac{1}{2}\log\left(\frac{\pi e}{2}\right)
	\leq C_\text{RBR,noncausal}
	\leq \bar{C}_\text{noncausal},
\end{equation}
where
\begin{equation}
	\bar{C}_\text{noncausal}
	=\sum_{k=1}^{\infty}p^2(1-p)^{k-1}
		\frac{k}{2}\log(1+\bar{B}/k).
\end{equation}
\end{proposition}
\begin{proof}
See Appendix~\ref{sec:bounds_noncausal}.
\end{proof}
In Appendix~\ref{sec:bounds_noncausal}, we consider another lower bound for the capacity with noncausal energy arrival information at the transmitter, which is tighter than the one in \eqref{eq:bounds_noncausal}, and plot it in Figure~\ref{fig:noncausal_causal} together with the upper bound in~\eqref{eq:bounds} on the capacity with causal energy arrival information.
It is clear from the graph that for some values of $\bar{B}$, the noncausal capacity is strictly greater than the causal capacity, further illustrating the observation we state in Proposition~\ref{prop:noncausal_strictly_greater}. 

\begin{figure}[!t]
\centering
	\includegraphics[width=2.5in]{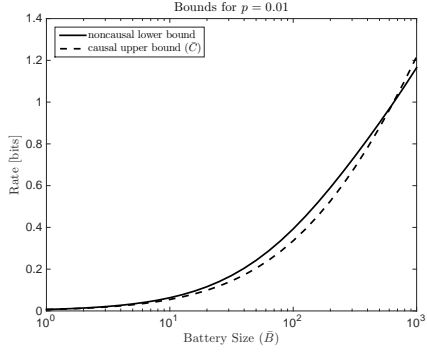}
\caption{Noncausal lower bound and causal upper bound for $p=0.01$. Noncausal capacity is strictly greater than causal capacity for some values of $\bar{B}$.}
\label{fig:noncausal_causal}
\end{figure}

\section{Capacity of the RBR Channel: Proof of Theorem~\ref{thm:equivalence}}
\label{sec:proof_mainthm}
Recall Definition~\ref{def:clipping} of the $(N)$-clipping channel. This is a memoryless time-invariant channel with i.i.d. states known to the receiver. Now, assume there was a feedback link from the receiver to the transmitter that could feed back the channel state information to the transmitter in a \emph{strictly causal} fashion. Since feedback cannot increase the capacity of a memoryless channel, the capacity of this new clipping channel with the state information $L_i$ available strictly causally at the transmitter is the same as that of our original clipping channel in Definition~\ref{def:clipping}. We will consider this new clipping channel with strictly causal state information in the sequel and modify the definition of the encoding functions in~\eqref{eq:clp_encoding} to depend on the past $L_i$:
\begin{equation}
	\tilde{f}^{(N)}_i:\mathcal{M}\times\mathcal{L}^{i-1}
	\to\tilde{\mathcal{X}}^{(N)},
	\qquad i=1,\ldots,n,
\label{eq:clp_encoding_SI}
\end{equation}
such that the transmitted symbol at time $i$ is $\tilde{X}^{(N)}_i=\tilde{f}^{(N)}_i(w,L^{i-1})$. With this modification, we show that the clipping channel is equivalent to the RBR channel by showing first that $p\cdot C_\text{clp}\leq C_\text{RBR}$ and then $p\cdot C_\text{clp}\geq C_\text{RBR}$.
\subsection{$p\cdot C_\text{clp}\leq C_\text{RBR}$}
\label{subsec:EH_scheme}
We will show that if rate $R$ is achievable for the $(N)$-clipping channel for some $N$, then rate $pR$ is achievable over the RBR channel. Fix $\varepsilon>0$.
Let $R< C_\text{clp}^{(N)}$ be an achievable rate for the $(N)$-clipping channel.
Then there exists $n_0$ such that for every $n> n_0$ there exists an $(M,n)$ code with probability of error $P_e^{\text{clp},(N)}<p\varepsilon/2$ and $\frac{\log M}{n}\geq R-\varepsilon$.

Fix such $n>n_0$.
Consider the following $(M,n')$ code for the RBR channel.
At time $t$, $1\leq t\leq n'$, the transmitter has causal knowledge of the energy arrivals $E^t$.
Let $K_t$ denote the number of energy arrivals up to time $t$,
and denote by $\{T_i\}_{i=0}^{K_t-1}$ the energy arrival times,
i.e. the times for which $E_t=\bar{B}$.
Note that $K_t\in\{1,\ldots,t\}$ is a random quantity and is a function of $E^t$.
Since we assume $E_1=\bar{B}$ (cf. Section~\ref{sec:definitions}), we always have $T_0=1$.
For each $i=1,\ldots,{K_t-1}$, define $L_i=T_i-T_{i-1}$.
As $n'\to\infty$, these $L_i$'s will behave as i.i.d. geometric RVs with parameter $p$.
In order to communicate message $w$, upon observing $E^t$, the transmitter sends the following symbol at time $t$:
\[
	X_t(w,E^t)=\begin{cases}
		\tilde{X}^{(N)}_{K_t,t-T_{K_t-1}+1}(w,L^{K_t-1})\\
			\qquad\qquad\qquad\qquad,t-T_{K_t-1}+1\leq N\\
		0\\ \qquad\qquad\qquad\qquad,t-T_{K_t-1}+1>N
	\end{cases}
\]
where $\tilde{X}^{(N)}_{i,j}(w,L^{i-1})$ is the $j$-th element of the $i$-th symbol in the $(N)$-clipping channel code $\tilde{f}_i^{(N)}$, defined in~\eqref{eq:clp_encoding_SI}.
This codeword will satisfy the energy constraint~\eqref{eq:EH_constraint}, since for each~$t$,
\[
	B_t=\bar{B}-\sum_{j=1}^{t-T_{K_t-1}}|\tilde{X}^{(N)}_{K_t,j}|^2
	\geq|\tilde{X}^{(N)}_{K_t,t-T_{K_t-1}+1}|^2,
\]
because 
$\sum_{j=1}^{t-T_{K_t-1}+1}|\tilde{X}^{(N)}_{K_t,j}|^2\leq\sum_{j=1}^{N}|\tilde{X}^{(N)}_{K_t,j}|^2\leq\bar{B};$ 
a clipping channel code must satisfy the energy constraint~\eqref{eq:clp_constraint}.

At the end of time slot $n'$, the receiver observes $Y^{n'}$ and $E^{n'}$, 
and forms $\{T_i\}_{i=0}^{K_{n'}-1}$ and $\{L_i\}_{i=1}^{K_{n'}-1}$.
Define also $L_{K_{n'}}\triangleq {n'}+1-T_{K_{n'}-1}$.
Define for $i=1,\ldots,K_{n'}$:
\[
	\tilde{Y}^{(N)}_{ij}=
	\begin{cases}
		Y_{T_{i-1}+j-1}&,j\leq L_i\\
		0&,j>L_i
	\end{cases}
\]
if $L_i\leq N$, and $\tilde{Y}^{(N)}_{i}=0$ if $L_i>N$.
Similarly, set $\tilde{Y}^{(N)}_{i}=0$ for $i=K_{n'}+1,\ldots,n$ (we will see shortly that the probability that $K_{n'}<n$ should vanish).
The receiver then sets $\hat{W}=\hat{W}((\tilde{Y}^{(N)})^n,L^n)$, where $\hat{W}=\tilde{g}^{(N)}$, the decoding function of the clipping channel code defined in~\eqref{eq:clp_decoding}.

To analyze the probability of error, we separate into two cases: $K_{n'}<n$ and $K_{n'}\geq n$. When $n'$ is chosen appropriately as a function of $n$, $K_{n'}<n$ should occur with very small probability, and otherwise when $K_{n'}\geq n$, intuitively we have been able to transmit all symbols of the $(N)$-clipping channel codeword of length $n$ and the channel we induce from $(\tilde{X}^{(N)})^n$ to $(\tilde{Y}^{(N)})^n$ is similar to the $(N)$-clipping channel, therefore the decoder should give vanishing probability of error. Formally,
\begin{align}
	P_e^{RBR}&=\Pr(\hat{W}\neq W\ \cap\ K_{n'}\geq n)\nonumber\\*
		&\qquad+\Pr(\hat{W}\neq W\ \cap\ K_{n'}<n)\nonumber\\
	&\leq \Pr(\hat{W}\neq W\ \cap\ K_{n'}\geq n)+\Pr(K_{n'}<n).\label{eq:EH_Pe}
\end{align}

We start with the first term.
Consider an $(M,n)$ \mbox{$(N)$-clipping} channel code.
For a given message $w$ and clipping length realizations $l^n$, the input $(\tilde{x}^{(N)})^n$ is uniquely determined and  therefore also the transition probability $p((\tilde{y}^{(N)})^n|(\tilde{x}^{(N)})^n,l^n)$. Denote $P_\text{clp}^{(N)}(\mathcal{E}|w,l^n)$ as the probability of error.
The total probability of error for the code, over the $(N)$-clipping channel, is then
\[
P_e^{\text{clp},(N)}=\frac{1}{M}\sum_{w=1}^{M}\sum_{l^n}P(l^n)P_\text{clp}^{(N)}(\mathcal{E}|w,l^n),
\]
where $P(l^n)=\prod_{i=1}^{n}P(l_i)$ and $P(l_i)=p(1-p)^{l_i-1}$.
We choose $n$ large enough so that $P_e^{\text{clp},(N)}\leq p\varepsilon/2$.

Now going back to our RBR channel, consider a realization of the energy arrival process $e^{n'}$, and denote $P_\text{RBR}(\mathcal{E}|w,e^{n'})$ as the probability of error given this realization and a message $w$, using the scheme described above.
Then we can write the first term in~\eqref{eq:EH_Pe} as:
\begin{align*}
	\lefteqn{\Pr(\hat{W}\neq W\ \cap\ K_{n'}\geq n)}\\*
	&\qquad=\frac{1}{M}\sum_{w=1}^{M}
	\sum_{\substack{e^{n'}: k_{n'}\geq n}}P(e^{n'})P_\text{RBR}(\mathcal{E}|w,e^{n'}),
\end{align*}
where $k_{n'}$ denotes the number of energy arrivals in $e^{n'}$, and $P(e^{n'}) = p^{k_{n'}-1}(1-p)^{n'-k_{n'}}$ (recall that $e_1$ is fixed),
and $P_{\text{RBR}}(\mathcal{E}|w,e^{n'})$ is the probability of error for the RBR channel using the scheme described above.

For a sequence $e^{n'}$ with at least $n$ energy arrivals, 
we can define $L^n(e^{n'})=l^n$ as the sequence of lengths of the first $n$ epochs, i.e. the times between consecutive energy arrivals.
If there are exactly $n$ energy arrivals, so that the last epoch is undefined, we define it to be $l_n\triangleq n'-\sum_{i=1}^{n-1}l_i$, i.e. it lasts until then end of the sequence $e^{n'}$.
Then we have $P_\text{RBR}(\mathcal{E}|w,e^{n'})=P^{(N)}_\text{clp}(\mathcal{E}|w,L^n(e^{n'}))$,
and denoting $\lambda\triangleq\sum_{i=1}^{n}l_i$, we have
\[
	P(e^{n'})=P(e^{\lambda})P(e_{\lambda+1}^{n'})\\
	=\frac{1}{p}P(l^n)P(e_{\lambda+1}^{n'}).
\]
If $k_{n'}=n$, then it is understood that ${P(e_{\lambda+1}^{n'})=P(e_{n'+1}^{n'})=1}$.
We can write the above sum as
\begin{align*}
	\lefteqn{\sum_{\substack{e^{n'}: k_{n'}\geq n}}P(e^{n'})P_\text{RBR}(\mathcal{E}|w,e^{n'})}\\*
	&\qquad=\sum_{\substack{l^n:\\ \lambda\leq {n'}}}
		\sum_{\substack{e^{n'}:\\ L^n(e^{n'})=l^n}}P(e^{n'})P_\text{RBR}(\mathcal{E}|w,e^{n'})\\
	&\qquad=\sum_{\substack{l^n:\\ \lambda\leq {n'}}}
		\frac{1}{p}P(l^n)P^{(N)}_\text{clp}(\mathcal{E}|w,l^n)\sum_{\substack{e^{n'}:\\ L^n(e^{n'})=l^n}}P(e^{n'}_{\lambda+1})\\
	&\qquad=\sum_{\substack{l^n:\\ \lambda\leq {n'}}}
		\frac{1}{p}P(l^n)P^{(N)}_\text{clp}(\mathcal{E}|w,l^n)
			\sum_{e^{n'}_{\lambda+1}}P(e^{n'}_{\lambda+1})
		\\
	&\qquad\leq\sum_{l^n}\frac{1}{p}P(l^n)P^{(N)}_\text{clp}(\mathcal{E}|w,l^n),
\end{align*}
which gives
$\Pr(\hat{W}\neq W\ \cap\ K_{n'}\geq n)\leq \frac{1}{p}P_e^{\text{clp},(N)}\leq\varepsilon/2$.

The second term in~\eqref{eq:EH_Pe} can be bounded using the law of large numbers.
Since $K_{n'}=1+\sum_{i=1}^{{n'}-1}Z_i$ with $Z_i\sim\text{Bernoulli}(p)$,
\[
	\Pr(K_{n'}<n)=\Pr\left(\frac{1}{{n'}-1}\sum_{i=1}^{{n'}-1}Z_i
	<\frac{n-1}{{n'}-1}\right)\leq\frac{\varepsilon}{2},
\]
for ${n'}$ large enough if $\frac{n-1}{{n'}-1}<p$, say $\frac{n-1}{{n'}-1}=p(1-\varepsilon)$.
${n'}$ can be chosen large enough so that also $n\geq n_0$.
The rate of this code is
\[
	\frac{\log M}{n'}
	\geq p(1-\varepsilon)\frac{\log M}{n}
	\geq p(1-\varepsilon)(R-\varepsilon)
	=pR-\epsilon^\prime,
\]
where $\varepsilon^\prime=p\varepsilon(1+R-\varepsilon)$.
We showed that for ${n'}$ large enough we can obtain $P_e^{\text{RBR}}\leq\varepsilon$ with rate close to
$pR$, i.e. $pR$ is achievable for the RBR channel,
which implies ${p\cdot C^{(N)}_\text{clp}\leq C_\text{RBR}}$.
Since we showed this for arbitrary $N$, we conclude that $p\cdot C_\text{clp}=p\cdot\sup_{N\geq1}C_\text{clp}^{(N)}\leq C_\text{RBR}$.\qed


\subsection{$C_\text{RBR}\leq p\cdot C_\text{clp}$}
\label{subsec:clp_scheme}
We show that any achievable rate for the RBR channel is $\varepsilon$-achievable on the $(N)$-clipping channel for some $N$ large enough.
Fix $\varepsilon>0$.
Let $R< C_\text{RBR}$ be an achievable rate for the RBR channel.
Then there exists $n'_0$ such that for all integers $n'>n'_0$ there exists an $(M,n')$ code with $P_e^{\text{RBR}}<\varepsilon/2$ and $\frac{\log M}{n'}\geq R-\varepsilon$.

Fix $n'>n'_0$ and $N$.
We suggest the following $(M,n)$ coding scheme for the $(N)$-clipping channel:
consider time $i$, $1\leq i\leq n$.
For a given sequence of past clipping lengths $L^{i-1}$,
denote the arrival times ${T}_{j}=1+\sum_{k=1}^{j}L_k$ for $j=0,1,\ldots,i-1$.
Define the sequence $E^{(i),n'}=({E}_1^{(i)},\ldots,{E}_{n'}^{(i)})$ as follows
(this sequence is different for different $i$'s):
\begin{equation}
	{E}_\ell^{(i)}=
	\begin{cases}
		\bar{B}&,\ell=T_0,{T}_1,{T}_2,\ldots,{T}_{i-1}\\
		0&,\text{otherwise}
	\end{cases}
	\label{eq:def_E(L)}
\end{equation}
Note that $E_\ell^{(i)}=0$ for all $\ell>T_{i-1}$.
It is possible for some values of ${T}_{j}$ to be larger than $n'$ -- these values are simply ignored.
At time $i$, the transmitter sends the codeword $\tilde{X}_i^{(N)}$, whose $j$-th component, $1\leq j\leq N$, is defined as
\[
	\tilde{X}^{(N)}_{ij}(w,L^{i-1})=
	\begin{cases}
		X_{{T}_{i-1}+j-1}(w,{E}^{(i),T_{i-1}+j-1})\\
			\qquad\qquad\qquad,{T}_{i-1}+j-1\leq n'\\
		0\\
			\qquad\qquad\qquad,{T}_{i-1}+j-1>n'
	\end{cases}
\]
where $X_t=f_t$, the encoding functions of the
RBR channel code, defined in~\eqref{eq:EH_encoding}.
Since $E^{(i)}_{T_{i-1}}=\bar{B}$ and $E^{(i)}_\ell=0$ for $\ell>T_{i-1}$,
by the definition of the RBR channel energy constraints~\eqref{eq:EH_constraint} and~\eqref{eq:EH_battery}, it is evident that $\|\tilde{X}^{(N)}_i(w,L^{i-1})\|^2\leq\bar{B}$.

The receiver observes $\tilde{Y}^{(N)}_i$ and $L_i$, and forms a concatenation of the non-clipped outputs.
More precisely, if $L_i\leq N$, set
$Y_{{T}_{i-1}+j-1}=\tilde{Y}_{ij}^{(N)}$ for $j=1,\ldots,L_i$.
If $L_i>N$, the receiver declares an error for the entire transmission.
At the end of time $n$, the receiver decides on message
$\hat{W}=\hat{W}(Y^{n'},{E}^{(n+1),n'})$,
where $\hat{W}=g$, the decoding function of the RBR channel code, defined in~\eqref{eq:EH_decoding}.

The probability of error analysis here follows the same lines of Section~\ref{subsec:EH_scheme}.
Denote the event
\begin{equation}
	E_0=\{T_n\leq n'\}\cup\bigcup_{i=1}^{n}\{L_i>N\}.
\end{equation}
We separate between the case $E_0^c$, for which we can show that the decoder gives vanishing probability of error, and the case $E_0$, the probability of which can be upper bounded by $\varepsilon/2$.

We can bound the probability of error for the $(N)$-clipping channel as follows:
\begin{align}
	P_e^{\text{clp},(N)}&=\Pr(\hat{W}\neq W\cap E_0^c)
		+\Pr(\hat{W}\neq W\cap E_0)\nonumber\\
	&\leq\Pr(\hat{W}\neq W\ \cap\ E_0^c)+\Pr(E_0).\label{eq:clp_Pe}
\end{align}

To bound the first term, consider an $(M,n')$ RBR channel code.
For a given message $w$ and energy arrival process realization $e^{n'}$, the input $x^{n'}$ is uniquely determined and therefore also the transition probability $p(y^{n'}|x^{n'},e^{n'})$.
Denote $P_\text{RBR}(\mathcal{E}|w,e^{n'})$ as the probability of error, so that the total probability of error for the RBR channel is 
\[
	P_e^\text{RBR}=\frac{1}{M}\sum_{w=1}^{M}\sum_{e^{n'}}
		P(e^{n'})P_\text{RBR}(\mathcal{E}|w,e^{n'}),
\]
where $P(e^{n'})$ is the probability of the energy arrival process, as described in Section~\ref{subsec:EH_scheme}.
Choose ${n'}$ large enough so that $P_e^\text{RBR}\leq\varepsilon/2$.

The first term in~\eqref{eq:clp_Pe} can be written as
\begin{align*}
	\Pr(\hat{W}\neq W\cap E_0^c)
	=\frac{1}{M}\sum_{w=1}^{M}
	\sum_{\substack{l^n:\\ \sum_{i=1}^{n}l_i\geq n'\\ l_i\leq N\ \forall i}}
	P(l^n)P_\text{clp}^{(N)}(\mathcal{E}|w,l^n),
\end{align*}
where $P(l^n)$ is the probability of $n$ independent clipping lengths as in Section~\ref{subsec:EH_scheme},
and $P_\text{clp}^{(N)}(\mathcal{E}|w,l^n)$ is the probability of error for the clipping channel using the scheme described above.

For any given $l^n$, there is a corresponding sequence of length $\lambda=\sum_{i=1}^n{l_i}$, $e^{\lambda}\in\{0,\bar{B}\}^{\lambda}$, as defined by~\eqref{eq:def_E(L)}.
For $l^n$ such that $\sum_{i=1}^{n}l_i\geq n'$ and $l_i\leq N$, $i=1,\ldots,n$,
denote by $E^{n'}(l^n)$ the first $n'$ elements of the corresponding $e^{\lambda}$.
Then 
\[
	P_\text{clp}^{(N)}(\mathcal{E}|w,l^n)
	=P_\text{RBR}(\mathcal{E}|w,E^{n'}(l^n)).
\]
Note also that $E^{n'}(l^n)$ can have at most $n$ non-zero elements (i.e. energy arrivals).

For a fixed $w$, we write
\begin{align}
	\lefteqn{\sum_{\substack{l^n:\\ \sum_{i=1}^{n}l_i\geq n'\\ l_i\leq N\ \forall i}}P(l^n)
		P_\text{clp}^{(N)}(\mathcal{E}|w,l^n)}\nonumber\\*
	&\qquad=\sum_{\substack{e^{n'}:\\ \text{at most $n$ $\bar{B}$'s}}}
		\sum_{\substack{l^n:\\ \sum_{i=1}^{n}l_i\geq n'\\ l_i\leq N\ \forall i\\E^{n'}(l^n)=e^{n'}}}
		P(l^n)P_\text{RBR}(\mathcal{E}|w,e^{n'})\nonumber\\
	&\qquad=\sum_{\substack{e^{n'}:\\ \text{at most $n$ $\bar{B}$'s}}}
		P_\text{RBR}(\mathcal{E}|w,e^{n'})
		\sum_{\substack{l^n:\\ \sum_{i=1}^{n}l_i\geq n'\\ l_i\leq N\ \forall i\\E^{n'}(l^n)=e^{n'}}}
		P(l^n)\nonumber\\
	&\qquad\leq\sum_{\substack{e^{n'}:\\ \text{at most $n$ $\bar{B}$'s}}}
		P_\text{RBR}(\mathcal{E}|w,e^{n'})
		\sum_{\substack{l^n:\\ \sum_{i=1}^{n}l_i\geq n'\\E^{n'}(l^n)=e^{n'}}}
		P(l^n).
		\label{eq:sumln}
\end{align}
Now, observe that a sequence $e^{n'}$ with exactly $i$ energy arrivals
(with $e_1$ being the first energy arrival by assumption)
can be described by $i-1$ epoch lengths $\{\tilde{l}_j\}_{j=1}^{i-1}$, 
where $\tilde{l}_j$ is the time between the $j$-th and $(j+1)$-th energy arrivals.
The last epoch length is always $\tilde{l}_i=n'-\sum_{j=1}^{i-1}\tilde{l}_j$.
Then, a sequence $l^n$ satisfies $E^{n'}(l^n)=e^{n'}$ if and only if
\begin{align*}
	l_j&=\tilde{l}_j,\qquad j=1,\ldots,i-1,\\
	l_i&\geq\tilde{l}_i.
\end{align*}
For example, say $e^{9}=100010100$. 
Then any sequence $l^n$ for which $l_1=4$, $l_2=2$, and $l_3\geq3$ will satisfy $E^{9}(l^n)=e^9$.

Fix $e^{n'}$ with $i$ energy arrivals. Then, since $e_1=\bar{B}$ w.p. 1,
\[
	P(e^{n'})=p^{i-1}(1-p)^{n'-i}
	=\frac{1}{p}\prod_{j=1}^{i}p(1-p)^{\tilde{l}_j-1}
	=\frac{1}{p}\prod_{j=1}^iP(\tilde{l}_j),
\]
where $P(\tilde{l}_j)$ is the probability of a geometric RV with parameter $p$.
We can write the sum in~\eqref{eq:sumln} as
\begin{align*}
	\sum_{\substack{l^n:\\ \sum_{i=1}^{n}l_i\geq n'\\ E^{n'}(l^n)=e^{n'}}}P(l^n)
	&=\sum_{\substack{l^n:\\ l^{i-1}=\tilde{l}^{i-1}\\ l_i\geq\tilde{l}_i}}P(l^n)\\
	&=\prod_{j=1}^{i-1}P(\tilde{l}_j)\Pr(L\geq\tilde{l}_i)\\
	&=P(e^{n'}),
\end{align*}
where the last step is because for a geometric RV,
\[\Pr(L\geq \ell)=\sum_{k=\ell}^{\infty}p(1-p)^{k-1}=(1-p)^{\ell-1}=\frac{1}{p}\Pr(L=\ell).\]

Plugging this back into the original expression for the error probability:
\begin{align*}
	\lefteqn{\Pr(\hat{W}\neq W\ \cap\ E_0^c)}\\
	&\qquad\leq\frac{1}{M}\sum_{w=1}^{M}
		\sum_{\substack{e^{n'}:\\ \text{at most $n$ $\bar{B}$'s}}}
		P_\text{RBR}(\mathcal{E}|w,e^{n'})
		P(e^{n'})\\
	&\qquad\leq
		\frac{1}{M}\sum_{w=1}^{M}
		\sum_{e^{n'}}
		P_\text{RBR}(\mathcal{E}|w,e^{n'})
		P(e^{n'})\\
	&\qquad=
		P_e^{\text{RBR}}\\
	&\qquad\leq\varepsilon/2.
\end{align*}

The second term in~\eqref{eq:clp_Pe} can be upper bounded by applying the union bound and the law of large numbers:
\begin{align}
	\Pr(E_0)&\leq\Pr(T_n-1<n')+\sum_{i=1}^{n}\Pr(L_i>N)\nonumber\\
	&=\Pr\left(\frac{1}{n}\sum_{i=1}^{n}{L_i}<\frac{n'}{n}\right)+n(1-p)^N\nonumber\\
	&\leq\frac{\varepsilon}{4}+\frac{\varepsilon}{4}
	=\frac{\varepsilon}{2},
\end{align}
for $n$ large enough if $n'/n<\mathbb{E}L=\frac{1}{p}$,
say $\frac{n'}{n}=\frac{1-\varepsilon}{p}$,
and $N>\frac{\log(\varepsilon/4n)}{\log(1-p)}=\frac{\log(\varepsilon(1-\varepsilon)/4pn')}{\log(1-p)}$.
We get that there exists $N$ for which $P_e^{\text{clp},(N)}\leq\varepsilon$ for $n$ sufficiently large.
The rate of this $(M,n)$ code is
\[
	\frac{\log M}{n}
	=\frac{1-\varepsilon}{p}\frac{\log M}{n'}
	\geq\frac{1-\varepsilon}{p}(R-\varepsilon)
	=\frac{R}{p}-\varepsilon^\prime,
\]
where $\varepsilon^\prime=\varepsilon\frac{1+R-\varepsilon}{p}$.
Next, we have from Fano's inequality:
\begin{align*}
	\frac{R}{p}-\varepsilon'
	&\leq\frac{\log M}{n}\\
	&\leq\frac{1}{n}\frac{1}{1-\varepsilon}[I((\tilde{X}^{(N)})^n;(\tilde{Y}^{(N)})^n|L^n)+h_2(\varepsilon)]\\
		&\leq\frac{1}{1-\varepsilon}\left[C_\text{clp}^{(N)}+\frac{h_2(\varepsilon)}{n}\right]\\
		&\leq\frac{1}{1-\varepsilon}\left[C_\text{clp}+\frac{h_2(\varepsilon)}{n}\right].
\end{align*}
Taking $\varepsilon\to0$, and observing that this applies to any ${R<C_\text{RBR}}$, gives $C_\text{RBR}/p\leq C_\text{clp}$.\qed

\section{Noncausal Side Information Strictly Increases Capacity: Proof of Proposition~\ref{prop:noncausal_strictly_greater}}
\label{sec:noncausal_strictly_greater}

We would like to show that $C_\text{RBR}<C_\text{RBR,noncausal}$, where $C_\text{RBR}$ and $C_\text{RBR,noncausal}$ are given by~\eqref{eq:EH_capacity} and~\eqref{eq:EH_capacity_noncausal} respectively.

Let $g_n(x^n)$ and $f_k(x^k)$ be the maximizing distributions in~\eqref{eq:EH_capacity} and~\eqref{eq:EH_capacity_noncausal} respectively, that is:
\begin{align}
	g_n(x^n)&=\argmax_{p(x^n):\ \|X^n\|^2\leq\bar{B}}
		\left\{\sum_{k=1}^{n}(1-p)^{k-1}
		I(X^k;Y^k)\right\},\label{eq:g_definition}\\
	f_k(x^k)&=\argmax_{p(x^k):\ \|X^k\|^2\leq\bar{B}}
		I(X^k;Y^k).
\end{align}
The maximizing distributions $f_k(x^k)$ are unique and are found explicitly in \cite{Smith1971,ShamaiBarDavid1995,
ChanHranilovicKschischang2005}.
For any $1\leq l\leq k$, let $f_k(x^l)$ denote the marginal distribution of $f_k(x^k)$, that is 
\[f_k(x^l)=\int f_k(x^k)dx_{l+1}^{k},\]
 and similarly for $g_n(x^l)$.
Denote $I(p(X^k))$ as the mutual information $I(X^k;X^k+Z^k)$ computed with input distribution $p(x^k)$.

Using this notation, we rewrite~~\eqref{eq:EH_capacity} and~\eqref{eq:EH_capacity_noncausal} as follows:
\begin{align*}
	C_\text{RBR}&=\lim_{n\to\infty}
		\sum_{k=1}^{n}p^2(1-p)^{k-1}
		I(g_n(X^k)),\\
	C_\text{RBR,noncausal}&=\lim_{n\to\infty}
		\sum_{k=1}^{n}p^2(1-p)^{k-1}
		I(f_k(X^k)).
\end{align*}
Obviously $C_\text{RBR}\leq C_\text{RBR,noncausal}$.
Suppose that also $C_\text{RBR}= C_\text{RBR,noncausal}$.
This will imply
\[
	\lim_{n\to\infty}\sum_{k=1}^{n}(1-p)^{k-1}
	[I(f_k(X^k))-I(g_n(X^k))]=0.
\]
Since each term in the sum is non-negative, we get in particular that the sum of the first two elements must vanish, or
\begin{align*}
	&\lim_{n\to\infty}\{I(g_n(X_1))+(1-p)I(g_n(X^2))\}\\
	&\hspace{2cm}=I(f_1(X_1))+(1-p)I(f_2(X^2)).
\end{align*}
Next, consider $g_2(x^2)$ as defined in~\eqref{eq:g_definition}.
Since $\|X^n\|^2\leq\bar{B}$ implies $\|X^2\|^2\leq\bar{B}$, we get for every $n\geq2$:
\begin{align}
	&I(g_n(X_1))+(1-p)I(g_n(X^2))\nonumber\\*
	&\qquad\leq\max_{p(x^2):\ \|X^2\|^2\leq\bar{B}}
		\big\{I(X_1;Y_1)+(1-p)I(X^2;Y^2)\big\}\nonumber\\
	&\qquad=
		I(g_2(X_1))+(1-p)I(g_2(X^2)),\nonumber\\
	\intertext{which implies}
	&I(f_1(X_1))+(1-p)I(f_2(X^2))\nonumber\\*
	&\qquad\leq
		I(g_2(X_1))+(1-p)I(g_2(X^2)).
	\label{eq:first2_ub}
\end{align}

Next, since $\|X^2\|^2\leq\bar{B}$ implies $|X_1|^2\leq\bar{B}$, we get that $I(g_2(X_1))\leq I(f_1(X_1))$.
Substituting in~\eqref{eq:first2_ub}, we get
$I(f_2(X^2))\leq I(g_2(X^2))$.
Since $f_2(x^2)$ is the unique maximizer of $I(X^2;Y^2)$, this implies $f_2(x^2)=g_2(x^2)$.
Substituting this back in~\eqref{eq:first2_ub}, we get $I(f_1(X_1))\leq I(g_2(X_1))$.
Again, from uniqueness, this implies $f_1(x)=g_2(x)$.
Together, we see that
\[
	f_1(x_1)=g_2(x_1)
	=\int g_2(x_1,x_2)dx_2
	=\int f_2(x_1,x_2)dx_2,
\]
which is a contradiction, since $f_1(x_1)$ is discrete \cite{Smith1971}, whereas $f_2(x_1,x_2)$ has discrete amplitude and uniform phase \cite{ShamaiBarDavid1995}.
Therefore we must have $C_\text{RBR}< C_\text{RBR,noncausal}$.

\section{Capacity Bounds: Proof of Proposition~\ref{prop:bounds}}
\label{sec:bounds}

In this section we will develop the upper and lower bounds to $C_{\text{RBR}}$, as given in Proposition~\ref{prop:bounds}.

\subsection{Upper Bound}
\label{subsec:upper_bound}

We can relax the energy constraint in~\eqref{eq:EH_capacity} to be only in expectation,
thus giving an upper bound:
\begin{align*}
	C_\text{RBR}
	&\leq \lim_{N\to\infty} \max_{\substack{p(x^N):\\ \mathbb{E}\|X^N\|^2\leq\bar{B}}}
		\sum_{k=1}^{N}p^2(1-p)^{k-1}I(X^k;X^k+Z^k)\\
	&\leq \lim_{N\to\infty}\max_{\substack{p(x^N):\\ \mathbb{E}\|X^N\|^2\leq\bar{B}}}
		\sum_{k=1}^{N}p^2(1-p)^{k-1}\sum_{i=1}^{k}I(X_i;X_i+Z_i)\\
	&= \lim_{N\to\infty}\max_{\substack{p(x^N):\\ \mathbb{E}\|X^N\|^2\leq\bar{B}}}\\*
		&\qquad\sum_{i=1}^{N}p(1-p)^{i-1} 
		[1-(1-p)^{N-i+1}]
		I(X_i;X_i+Z_i)\\
	&\leq \lim_{N\to\infty}\max_{\substack{p(x^N):\\ \mathbb{E}\|X^N\|^2\leq\bar{B}}}
		\sum_{i=1}^{N}p(1-p)^{i-1} 
		I(X_i;X_i+Z_i)\\
	&= \lim_{N\to\infty} \max_{\substack{\{\mathcal{E}_i\}_{i=1}^{N}:\\
					\mathcal{E}_i\geq0\ ,i=1,\ldots,N\\
					\sum_{i=1}^{N}\mathcal{E}_i\leq\bar{B}}}
		\sum_{i=1}^{N}p(1-p)^{i-1}\frac{1}{2}\log(1+\mathcal{E}_i),
\end{align*}
where the last equality is obtained by choosing $X_i\sim\mathcal{N}(0,\mathcal{E}_i)$
independent of each other.
This gives the RHS in~\eqref{eq:bounds}.\qed

\subsection{Lower Bound}
\label{subsec:lower_bound}

To lower bound~\eqref{eq:EH_capacity},
we can choose a suboptimal distribution for which the $X_i$'s are independent,
i.e. $p(x^N)=\prod_{i=1}^{N}p(x_i)$,
and each of them satisfies $|X_i|^2\leq\mathcal{E}_i$ a.s. for some 
$\mathcal{E}_i\geq0$.
To satisfy the total energy constraint we must have 
$\sum_{i=1}^{N}\mathcal{E}_i\leq\bar{B}$.
Under this input distribution, we have for every $i$: $I(X_i;Y_i)\leq\frac{1}{2}\log(1+\mathcal{E}_i)\leq\frac{1}{2}\log(1+\bar{B})$, and thus for every $N$:
\begin{align*}
	\hspace{1em}\lefteqn{\hspace{-1em}\sum_{k=1}^{N}p^2(1-p)^{k-1}I(X^k;X^k+Z^k)}\\
	&=\sum_{i=1}^{N}p(1-p)^{i-1}[1-(1-p)^{N-i+1}]I(X_i;X_i+Z_i)\\
	&\geq\sum_{i=1}^{N}p(1-p)^{i-1}I(X_i;X_i+Z_i)\\*
	&\hspace{4em}
		-N(1-p)^N\frac{1}{2}\log(1+\bar{B}).
\end{align*}
Taking $N\to\infty$, the second term vanishes, and we are left with the following lower bound:
\[
	C_\text{RBR}\geq\lim_{N\to\infty}\max_{\substack{\{\mathcal{E}_i\}_{i=1}^{N}:\\ \mathcal{E}_i\geq0\ ,i=1,\ldots,N\\ \sum_{i=1}^{N}\mathcal{E}_i\leq\bar{B}}}
		\sum_{i=1}^{N}p(1-p)^{i-1}I(X_i;X_i+Z_i).
\]
Since $p(x_i)$ was arbitrary, we can choose it to maximize  $I(X_i;X_i+Z_i)$. 
We obtain
\begin{equation}
\label{eq:EH_1dlowerbound}
	C_\text{RBR}\geq\lim_{N\to\infty}
	\max_{\substack{\{\mathcal{E}_i\}_{i=1}^{N}:\\ \mathcal{E}_i\geq0\ ,i=1,\ldots,N\\ \sum_{i=1}^{N}\mathcal{E}_i\leq\bar{B}}}
	\sum_{i=1}^{N}p(1-p)^{i-1}C_\text{Smith}(\mathcal{E}_i),
\end{equation}
where
\begin{equation}
\label{eq:smith_capacity_def}
C_\text{Smith}(\mathcal{E})\triangleq\max_{p(x):\ X^2\leq\mathcal{E}}I(X;Y)
\end{equation}
is the capacity of the amplitude constrained scalar AWGN channel studied in~\cite{Smith1971},
where the optimal value for this mutual information maximization problem is found. 
Unfortunately, it is not tractable.
Hence, as done in~\cite{DongOzgur2014,OzarowWyner1990},
we lower bound it as follows:
\begin{align}
	C_\text{Smith}(\mathcal{E})
	&\geq \frac{1}{2}\log\left(1+\frac{\mathcal{E}}{3}\right)
		-\frac{1}{2}\log\left(\frac{\pi e}{6}\right)\nonumber\\
	&\geq \frac{1}{2}\log(1+\mathcal{E})
		-\frac{1}{2}\log\left(\frac{\pi e}{2}\right).
		\label{eq:smith_lower_bound}
\end{align}
Plugging this into~\eqref{eq:EH_1dlowerbound} gives the LHS of~\eqref{eq:bounds}.



\bibliographystyle{IEEEtran}
\bibliography{IEEEabrv,energy_harvesting}


\appendices

\section{Noncausal Capacity: Proof of Theorem~\ref{thm:capacity_noncausal}}
\label{sec:noncausal_capacity_proof}

With noncausal energy arrival information, the RBR channel model remains the same except for equation~\eqref{eq:EH_encoding}, which becomes
\[
	f_t:\mathcal{M}\times \mathcal{E}^n\to\mathcal{X}
		,\qquad t=1,\ldots,n.
\]

The $(N)$-clipping channel now has causal knowledge of the clipping lengths at the transmitter.
It is a well known fact that with i.i.d. side information available at both the transmitter and the receiver, the capacity is the same whether it is available causally or noncausally.
Therefore we can repeat the steps of \Fref{sec:proof_mainthm}, while altering the definition of the $(N)$-clipping channel by replacing~\eqref{eq:clp_encoding_SI} with
\[
	\tilde{f}^{(N)}_i:\mathcal{M}\times\mathcal{L}^{n}
	\to\tilde{\mathcal{X}}^{(N)},
	\qquad i=1,\ldots,n,
\]
concluding that
\begin{equation}
	C_\text{RBR,noncausal}=p\cdot C_\text{clp,noncausal}.
\end{equation}

The capacity of a memoryless channel with i.i.d. side information available at both the transmitter and the receiver is given by (see e.g.~\cite[Chapter 7]{ElGamalKim2011})
\[
	C_\text{clp,noncausal}^{(N)}=
	\max_{\substack{p(\tilde{x}^{(N)}|l):\\ \|\tilde{X}^{(N)}\|^2\leq\bar{B}}} I(\tilde{X}^{(N)};\tilde{Y}^{(N)}|L).
\]
Then, following the same steps leading to~\eqref{eq:clp_capacity} in Section~\ref{sec:main_results}, we get
\begin{align*}
	\hspace{2em}\lefteqn{\hspace{-2em}C_\text{clp,noncausal}^{(N)}}\\
	&=\sum_{k=1}^{\infty}p(1-p)^{k-1}
		\max_{\substack{p(\tilde{x}^{(N)}|l=k):\\
			\|\tilde{X}^{(N)}\|^2\leq\bar{B}}}
		I(\tilde{X}^{(N)};\tilde{Y}^{(N)}|L=k)\\
	&=\sum_{k=1}^{N}p(1-p)^{k-1}
		\max_{\substack{p(x^k):\\ \|X^k\|^2\leq\bar{B}}}
		I(X^k;X^k+Z^k).
\end{align*}
Finally, we have
\begin{align*}
	C_\text{RBR,noncausal}&=p\cdot\lim_{N\to\infty}
		C_\text{clp,noncausal}^{(N)}\\
	&=\sum_{k=1}^{\infty}p^2(1-p)^{k-1}
		\max_{\substack{p(x^k):\\ \|X^k\|^2\leq\bar{B}}}
		I(X^k;X^k+Z^k).
\end{align*}
which is the expression in~\eqref{eq:EH_capacity_noncausal}.\qed

\section{Noncausal Capacity Bounds: Proof of Proposition~\ref{prop:bounds_noncausal}}
\label{sec:bounds_noncausal}

We prove Proposition~\ref{prop:bounds_noncausal} following the same steps as in the previous sections.
For the upper bound, we can similarly relax the energy constraint in~\eqref{eq:EH_capacity_noncausal} to be only in expectation,
thus giving an upper bound:
\begin{align*}
	C_\text{RBR,noncausal}
	&\leq\sum_{k=1}^{\infty}p^2(1-p)^{k-1}
		\max	_{\substack{p(x^k):\\ 
			\mathbb{E}\|X^k\|^2\leq\bar{B}}}
		I(X^k;X^k+Z^k)\\
	&=\sum_{k=1}^{\infty}p^2(1-p)^{k-1}
		\frac{k}{2}\log(1+\bar{B}/k),
\end{align*}
where the last line is a known result for vector Gaussian channels.

We derive a lower bound for~\eqref{eq:EH_capacity_noncausal} by considering a suboptimal input distribution.
Let $p(x^k)=\prod_{i=1}^{k}p_k(x_i)$, where $p_k(x)$ is some distribution for which $X^2\leq\bar{B}/k$ a.s.
We then have the following lower bound:
\begin{equation}
\label{eq:noncausal_lower_bound_tight}
	C_\text{RBR,noncausal}\geq
	\sum_{k=1}^{\infty}p^2(1-p)^{k-1}kC_\text{Smith}
		(\bar{B}/k),
\end{equation}
where $C_\text{Smith}(\mathcal{E})$ is defined in~\eqref{eq:smith_capacity_def}.
The expression has been evaluated numerically using the algorithm suggested in~\cite{Smith1971} and plotted in Figure~\ref{fig:noncausal_causal}.

Next, using~\eqref{eq:smith_lower_bound}, we can further lower bound~\eqref{eq:noncausal_lower_bound_tight} to obtain~\eqref{eq:bounds_noncausal}.

\section{Optimal Online Power Control}
\label{sec:kkt_solution}


We solve the maximization problem in~\eqref{eq:def_Cbar}.
Writing the problem in standard form and using KKT conditions, we have for $i=1,\ldots,N$:
\[
	-p(1-p)^{i-1}\frac{1}{2}\frac{1}{1+\mathcal{E}_i}-\lambda_i+\tilde{\lambda}=0,
\]
with $\lambda_i,\tilde{\lambda}\geq0$ and the complementary slackness conditions:
$\lambda_i\mathcal{E}_i=0$ and $\tilde{\lambda}(\sum_{i=1}^{N}\mathcal{E}_i-B)=0$.

To obtain the non-zero values of $\mathcal{E}_i$, we set $\lambda_i=0$:
\begin{equation}\label{eq:Esolution}
	\mathcal{E}_i=\frac{p(1-p)^{i-1}}{2\tilde{\lambda}}-1.
\end{equation}
Since $\mathcal{E}_i\geq0$, this implies $\tilde{\lambda}\leq\frac{p(1-p)^{i-1}}{2}$ for all $i$ for which $\mathcal{E}_i>0$.
This is a decreasing function of $i$, therefore there exists an integer $\tilde{N}$ such that $\mathcal{E}_i>0$ for $i=1,\ldots,\tilde{N}$ and $\mathcal{E}_i=0$ for $i>\tilde{N}$.
Since we are solving this problem for $N\to\infty$, it is safe to assume $\tilde{N}<N$.

Next we apply the total energy constraint (which must hold with equality, since increasing $\mathcal{E}_i$ for any $i$ will only increase the objective):
\[
	\bar{B}=\sum_{i=1}^{N}\mathcal{E}_i
	=\sum_{i=1}^{\tilde{N}}
		\left(\frac{p(1-p)^{i-1}}
		{2\tilde{\lambda}}-1\right)
	=\frac{1-(1-p)^{\tilde{N}}}{2\tilde{\lambda}}
		-\tilde{N}
\]
\begin{equation}\label{eq:lambda_tilde}
	\tilde{\lambda}=\frac{1-(1-p)^{\tilde{N}}}
	{2(\bar{B}+\tilde{N})}.
\end{equation}
Since $\tilde{\lambda}\leq\frac{p(1-p)^{i-1}}{2}$ for all $i=1,\ldots,\tilde{N}$, we must have $\tilde{\lambda}>\frac{p(1-p)^{\tilde{N}}}{2}$:
\[
	\frac{1-(1-p)^{\tilde{N}}}{2(\bar{B}+\tilde{N})}
	> \frac{p(1-p)^{\tilde{N}}}{2}
\]
\[
	1 > (1-p)^{\tilde{N}}[1+p(\bar{B}+\tilde{N})],
\]
which implies $\tilde{N}$ is the smallest integer satisfying this inequality.\qed

\end{document}